\documentclass[journal,twoside,web]{ieeecolor}
\usepackage{lcsys}
\usepackage{cite}
\usepackage{amsmath,amssymb,amsfonts}
\usepackage{algorithmic}
\usepackage{graphicx}
\usepackage{textcomp}

\usepackage{multirow}
\usepackage{subcaption}
\usepackage{xcolor}

\usepackage{epsfig} 
\usepackage{times} 
\usepackage{nicefrac, mathtools}
\usepackage{hyperref}
\usepackage{url}
\usepackage[capitalize]{cleveref}

\usepackage{scalerel}
\usepackage{graphicx}
\graphicspath{{./figures/}}
\usepackage{placeins}
\usepackage{xcolor}
\usepackage{bm}

\usepackage{newtxtext, newtxmath}

\usepackage{tikz, tikzscale,pgfplots}
\usetikzlibrary{patterns}
\usetikzlibrary{external}
\usepgfplotslibrary{fillbetween}

\pgfplotsset{width=10\columnwidth /10, compat = 1.13, 
	height = 55\columnwidth /100, grid= major, 
	legend cell align = left, ticklabel style = {font=\scriptsize},
	every axis label/.append style={font=\small},
	legend style = {font=\footnotesize},title style={yshift=-7pt, font = {\small\sffamily}} }

\def\BibTeX{{\rm B\kern-.05em{\sc i\kern-.025em b}\kern-.08em
    T\kern-.1667em\lower.7ex\hbox{E}\kern-.125emX}}
\newtheorem{assumption}{Assumption} 
\newtheorem{definition}{Definition} 
\newtheorem{lemma}{Lemma}
\newtheorem{proposition}{Proposition}

\newtheorem{theorem}{Theorem}

\begin{document}
\title{Distributed Risk-Sensitive Safety Filters for Uncertain Discrete-Time Systems}
\author{Armin Lederer, \IEEEmembership{Member, IEEE}, Erfaun Noorani, \IEEEmembership{Member, IEEE}, and Andreas Krause, \IEEEmembership{Fellow, IEEE}
\thanks{This work was supported as a part of NCCR Automation, a National Centre of Competence in Research, funded by the Swiss National Science Foundation (grant number 51NF40 225155).}
\thanks{A. Lederer is with the Department of Electrical and Computer Engineering, College of Design and Engineering, National University of Singapore, Singapore (email: armin.lederer@nus.edu.sg)}%
\thanks{E. Noorani can be contacted at enoorani@umd.edu.}%
\thanks{A. Krause is with the Learning \& Adaptive Systems Group, Department of Computer Science, ETH Zurich, Switzerland (email:  krausea@inf.ethz.ch).}%
}

\maketitle
\thispagestyle{empty}

\begin{abstract}
Ensuring safety in multi-agent systems is a significant challenge, particularly in settings where centralized coordination is impractical. In this work, we propose a novel risk-sensitive safety filter for discrete-time multi-agent systems with uncertain dynamics that leverages control barrier functions (CBFs) defined through value functions. Our approach relies on centralized risk-sensitive safety conditions based on exponential risk operators to ensure robustness against model uncertainties. We introduce a distributed formulation of the safety filter by deriving two alternative strategies: one based on worst-case anticipation and another on proximity to a known safe policy. By allowing agents to switch between strategies, feasibility can be ensured. Through detailed numerical evaluations, we demonstrate the efficacy of our approach in maintaining safety without being overly conservative.
\end{abstract}

\begin{IEEEkeywords}
Distributed control, 
uncertain systems,
constrained control,
data-driven control.
\end{IEEEkeywords}

\section{Introduction}
\label{sec:introduction}
\IEEEPARstart{M}{ulti-agent} systems (MASs) can be used to model a variety of applications with possible scenarios ranging from cooperative manipulation, where multiple robots jointly move an object they physically interact with, to robot swarms, in which large numbers of autonomous 
robots operate in a shared space. This has led to a growing interest in distributed reinforcement learning (RL) techniques in recent years to optimize the performance of MASs \cite{gronauer_multi-agent_2022}.
However, the assurance of safety requirements, e.g., restrictions on interaction forces in cooperative manipulation or collision constraints between agents in swarms, remains a challenging problem.

In single-agent systems, safety filters are an effective tool for rendering nominal policies safe \cite{wabersich_data-driven_2023}. Control barrier functions (CBFs) \cite{Ames2017} are a particular form of safety filters, which are frequently used due to their comparatively low computational complexity.  
While there is a rich literature on CBFs for continuous-time, control-affine systems including topics such as robustness against model uncertainties and learning of CBFs \cite{cohen2024safety}, these approaches do not directly extend to the discrete-time setting \cite{agrawal_discrete_2017}. In this setting, it has been shown that CBFs can be effectively learned from roll-out trajectories \cite {so_how_2024},
which can be combined with the optimization of trajectories via RL \cite{Curi2022}. Although CBFs have been developed for deterministic systems, they can be adapted to the stochastic setting~\cite{cosner_robust_2023}. Risk-sensitive formulations offer particular advantages in this setting as they allow to tune the aversion against uncertainty due to stochasticity, which can reduce the conservatism of safety conditions \cite{Ahmadi2022}. This can be particularly beneficial when learning CBFs for stochastic systems with uncertain dynamics from trajectories \cite{lederer_risk-sensitive_2023}.

For continuous-time multi-agent systems, many of the results for single-agent CBFs can be straightforwardly adapted by exploiting the linearity of CBF-based safety conditions \cite{GARG2024100948}. However, this is generally not possible for discrete-time dynamical systems \cite{ahmadi_safe_2019}. Therefore, CBFs must be specifically designed to allow a distributed implementation, e.g., for collisions constraints \cite{cheng_safe_2020}. While roll-out based approaches for learning CBFs can be employed to automate the design of CBFs for deterministic dynamics with perfect knowledge of the MAS in a centralized approach \cite{zhang_discrete_2025}, the resulting barrier functions require a locally coordinated computation of control inputs. 
This requirement makes it challenging to apply CBFs in distributed safety filters that run independently on each agent. Thus, the design of distributed safety filters remains an open challenge, especially when accounting for process noise and uncertain system dynamics.\looseness=-1

In this paper, we propose a novel approach for distributed risk-sensitive safety filters for uncertain discrete-time dynamical systems which exploits CBFs defined through value functions. Our contributions are threefold: 
\begin{itemize}
    \item \textbf{Learning risk-sensitive CBF conditions for safety:} We derive centralized risk-sensitive safety conditions that employ exponential risk operators to CBFs instead of  expectations. By expressing the CBF through optimal value functions, our approach naturally lends itself to obtaining CBFs via off-the-shelf RL methods.
    \item \textbf{Distributed formulation of safety filter:} We derive two adaptations of the centralized CBF conditions to the distributed setting in order to enable their independent evaluation on each agent in a distributed safety filter. We propose a switching mechanism for selecting the employed safety condition which ensures the feasibility and limits the conservatism of our distributed safety filter.
    \item \textbf{Detailed numerical evaluation:} We provide a detailed numerical analysis of the parameter dependency of the proposed safety filter and compare it to a centralized baseline. The favorable properties of our approach are demonstrated in two scenarios which also provide insights into the filter's behavior with growing numbers of agents.
\end{itemize}

\section{Problem Setting}\label{sec:problem}

We consider a MAS with $M\in\mathbb{N}$ agents,
where the dynamics of a single agent are described by
\begin{align}\label{eq:dynamics_individual}
    \bm{x}_{k+1}^i=\bm{f}^i(\bm{x}_k,\bm{u}_k^i,\bm{\omega}_k^i).
\end{align}
{The overall system state $\bm{x}_k=[(\bm{x}_k^1)^T,\ldots,(\bm{x}_k^M)^T]^T$ is the concatenation of individual agent states $\bm{x}^i_k\in\mathbb{R}^{d_x}$, $\bm{u}_k^i\in\mathbb{R}^{d_u}$
\parfillskip=0pt
\newline 
$~$

\vspace{-0.3cm}
\noindent 
is the control input for the $i$-th agent, $\bm{\omega}_k^i\in\mathbb{R}^{d_\omega}$ denotes the process noise of agent $i$ sampled from some zero-mean 
distribution $\rho$, and $\bm{f}^i:\mathbb{R}^{Md_x}\times\mathbb{R}^{d_u}\times\mathbb{R}^{d_\omega}\rightarrow \mathbb{R}^{d_x}$ is the transition function of agent $i$, which we assume to be unknown.
We denote the overall noise, control input, and dynamics by $\bm{\omega}_k=[(\bm{\omega}_k^1)^T,\ldots,(\bm{\omega}_k^M)^T]^T$ and $\bm{u}_k=[(\bm{u}_k^1)^T,\ldots,(\bm{u}_k^M)^T]^T$, and $\bm{f}(\cdot,\cdot,\cdot) = [(\bm{f}^1(\cdot,\cdot,\cdot))^T,\ldots, (\bm{f}^M(\cdot,\cdot,\cdot))^T]^T$, respectively. 
The dependency of each agent's individual dynamics $\bm{f}^i(\cdot,\cdot,\cdot)$ on the overall state $\bm{x}_k$ allows couplings between agents.}

We assume that agents exchange their states $\bm{x}_k^i$ with each other, such that each agent $i$ knows the global state $\bm{x}_k$, but the action $\bm{u}_k^i$ needs to be determined independently; 
no additional communication iteration is available for the computation of each action. 
For determining the local actions $u_k^i$, we consider given nominal policies $\bm{\pi}_{\mathrm{nom}}^i:\mathbb{R}^{Md_x}\rightarrow\mathbb{R}^{d_u}$, which are designed to execute a desired task. These policies can be obtained, e.g., by solving a joint optimal control problem for all agents 
using centralized  or distributed RL techniques~\cite{gronauer_multi-agent_2022}.

In addition to the execution of a desired task, safety commonly needs to be ensured by the employed policies. We consider here safety in terms of state constraints $\bm{x}_k\in\mathbb{X}_{\mathrm{safe}}, \forall k\in\mathbb{N}$
for some subset $\mathbb{X}_{\mathrm{safe}}\subset\mathbb{R}^{M d_x}$. Note that this safety notion allows the expression of inter-agent constraints, e.g., to avoid collisions between agents, and individual agents' safety requirements, e.g., to avoid crashes 
with obstacles. 
Due to the potentially unbounded process noise $\bm{\omega}_k$ in \eqref{eq:dynamics_individual}, the safety constraint $\bm{x}_k\in\mathbb{X}_{\mathrm{safe}}, \forall k\in\mathbb{N}$ can generally not be ensured deterministically. Therefore, we resort to the following probabilistic notion of safety.\looseness=-1
\begin{definition}
    A policy $\bm{\pi}$ is $K$-step $\delta$-safe for a state $\bm{x}\in\mathbb{R}^{Md_x}$ if it holds that $P(\bm{x}_{k}\in\mathbb{X}_{\mathrm{safe}},~\forall k=0,\ldots,K|\bm{x}_{0}=\bm{x})\geq 1-\delta$, where states $\bm{x}_{k}$ are defined in \eqref{eq:dynamics_individual}.
\end{definition}

{A common way to ensure the satisfaction of such a safety constraint is through the concept of CBFs. A CBF is a continuous function $h:\mathbb{R}^{Md_x}\rightarrow\mathbb{R}$ whose 0-superlevel set $\mathbb{H}_0$ is contained in the safe set $\mathbb{X}_{\mathrm{safe}}$, i.e., $\mathbb{H}_0=\{\bm{x}\in\mathbb{R}^{Md_x}: h(\bm{x})\geq 0  \}\subset \mathbb{X}_{\mathrm{safe}}.$ If $h(\cdot)$ satisfies
\begin{align}
    \mathbb{E}\left[h(\bm{f}(\bm{x},\bm{\pi}(\bm{x}),\bm{\omega}))\right] &\geq \alpha h(\bm{x})+\epsilon \label{eq:h_cond_noisy}
\end{align}
for all $\bm{x}\in\mathbb{H}_0$, $K$-step $\delta$-safety for $\pi(\cdot)$ can be shown under weak assumptions \cite{cosner_robust_2023}. 
Due to \eqref{eq:h_cond_noisy}, we can immediately define a centralized safety filter through the optimization problem
\begin{subequations}\label{eq:safety_filter_default}
    \begin{align}
    \bm{u}_{\mathrm{safe}}=&\min_{\bm{u}\in\mathbb{U}^M}\|\bm{u}-\bm{\pi}_{\mathrm{nom}}(\bm{x})\|^2\\
    &\text{s.t. } \mathbb{E}[h(\bm{f}(\bm{x},\bm{u},\bm{\omega}))]\geq \alpha h(\bm{x})+\epsilon\label{eq:CBF_constraint}
    \end{align}  
\end{subequations}
where 
$\mathbb{U}\subset\mathbb{R}^{d_u}$ denotes control constraints for a single agent. 
Note that this centralized filter needs a coordinated computation of control inputs and thereby requires additional communication compared to a pure exchange of state information. 
Despite being conceptionally simple, a core challenge of \eqref{eq:safety_filter_default} lies in finding a suitable CBF for a safe set $\mathbb{X}_{\mathrm{safe}}$, which ensures the feasibility of the constraint \eqref{eq:CBF_constraint}. Moreover, this approach assumes exact knowledge of the dynamics $f(\cdot,\cdot,\cdot)$, while we consider the system dynamics to be unknown. Since we cannot expect to ensure safety without any prior knowledge of the system behavior, we make the following assumption.\looseness=-1  
\begin{assumption}\label{ass:prob_model}
A probability distribution $\mathcal{F}$ over possible dynamics $\bm{f}$ is known, i.e., $\bm{f}\sim\mathcal{F}$.
\end{assumption}}
This assumption covers a wide range of probabilistic models that are employed in practice for learning system dynamics. For example, it includes learning methods such as Gaussian process regression \cite{Rasmussen2006} and deep learning-based approaches 
\cite{rothfuss_bridging_2024}. 
Moreover, parametric uncertainty can be expressed in the form required by \cref{ass:prob_model}, which underlines its flexibility.

To address the model uncertainty imposed by \cref{ass:prob_model}, we consider the problem of deriving a risk-sensitive safety filter with a structure analogous to \eqref{eq:safety_filter_default}, which allows a distributed execution on each agent individually. The goal is the design of an algorithm for learning a CBF as the basis for this safety filter given the distribution $\mathcal{F}$ and the safe set $\mathbb{X}_{\mathrm{safe}}$.\looseness=-1

\section{Distributed Safety Filters based on CBFs}\label{sec:safety_filter}

We start our derivations by developing a centralized approach in \cref{subsec:CBF_cent} that allows the formulation of CBFs through value functions, such that RL techniques can be employed for learning safety filters. Based on these CBFs, we first propose an approach for the worst-case evaluation of safety conditions without access to other agents' actions in \cref{subsec:worst_case}. To overcome the lack of feasibility guarantees of this strategy, we develop a filter that ensures safety through the proximity to a known safe policy in \cref{subsec:control_prox}. In \cref{subsec:merging}, we finally propose an approach that guarantees both safety and feasibility by switching between them.\looseness=-1

\subsection{Centralized Learning of Control Barrier Functions}\label{subsec:CBF_cent}

Similar as in our prior work \cite{Curi2022, lederer_risk-sensitive_2023}, we leverage the connection between value and barrier functions to enable the application of RL techniques for determining barrier functions. For this purpose, we consider the value function $V_{\bm{\pi}}(\bm{x})=\mathbb{E}[ \sum_{k=1}^\infty \gamma^k c(\bm{x}_k) ]$
for a cost function $c:\mathbb{R}^{M d_x}\rightarrow \mathbb{R}_{0,+}$ and a policy $\bm{\pi}(\cdot)$, such that $\bm{x}_{k+1}$ is defined via \eqref{eq:dynamics_individual}. Note that this function is well-defined by taking the expectation with respect to both the noise distribution $\rho$ and the distribution of the unknown function $\mathcal{F}$. The sub-level set $\mathbb{V}_{\bm{\pi}}^{\xi}=\{\bm{x}\in\mathbb{R}^{Md_x}: V_{\bm{\pi}}(\bm{x})\leq \xi \}$
of the value function $V_{\bm{\pi}}(\cdot)$ for a threshold $\xi\in\mathbb{R}$  can be used to inner-approximate the safe set $\mathbb{X}_{\mathrm{safe}}$ under weak assumptions as we show in the following lemma.
\begin{lemma}[\cite{Curi2022}]\label{lem:set2cost}
Assume there exists a constant $\hat{c}\in\mathbb{R}_+$, such that the cost $c:\mathbb{R}^{Md_x}\rightarrow\mathbb{R}_{0,+}$ satisfies
\begin{align}\label{eq:c_cond}
    c(\bm{x})\geq \hat{c}\quad \forall\bm{x}\in\overline{\mathbb{X}}_{\mathrm{safe}}=\mathbb{R}^{Md_x}\setminus\mathbb{X}_{\mathrm{safe}}.
\end{align} 
Then, there exists a constant $\xi\in\mathbb{R}_+$, such that the intersection between 
$\mathbb{V}_{\bm{\pi}}^{\xi}$ and $\overline{\mathbb{X}}_{\mathrm{safe}}$ is empty, i.e., $\mathbb{V}_{\bm{\pi}}^{\xi}\cap\overline{\mathbb{X}}_{\mathrm{safe}}=\emptyset$. 
\end{lemma}
Even though \eqref{eq:c_cond} is sufficient for allowing an inner-approximation of $\mathbb{X}_{\mathrm{safe}}$ using a sub-level set $\mathbb{V}_{\bm{\pi}}^{\xi}$, the particular choice of $c(\cdot)$ can have an impact on the approximation quality. Thus, we will focus on continuous extensions of the indicator function, for which it is straightforward to see that they are well-suited for such approximations.

In the sequel, let $\xi$ denote a threshold such that the sub-level set $\mathbb{V}_{\bm{\pi}}^{\xi}$ and $\overline{\mathbb{X}}_{\mathrm{safe}}$ do not intersect, which can be found, e.g., via a robust optimization formulation \cite{lederer_risk-sensitive_2023}. Then, we can define 
\begin{align}\label{eq:barrier}
    h(\bm{x})=\xi-V_{\bm{\pi}}(\bm{x})
\end{align}
as a barrier function, which satisfies $\mathbb{H}_0\subset \mathbb{X}_{\mathrm{safe}}$ by definition. To allow a targeted handling of uncertainty and the tuning of the risk aversion, we do not adapt \eqref{eq:h_cond_noisy}. Instead,
we follow a different strategy and derive a risk-sensitive safety condition for our considered setting with uncertain dynamics $\bm{f}(\cdot,\cdot,\cdot)$. Our approach is based on the exponential risk operator \cite{286253}, which is defined as $\mathbb{R}_{\beta}[C]=\frac{1}{\beta} \log\left(\mathbb{E}\left[ \exp\left( \beta C \right)  \right]\right)$
for a random variable $C$ and risk parameter $\beta\in\mathbb{R}_+$. 
We exploit the design freedom resulting from the risk parameter $\beta$ to derive the following safety condition with tunable risk aversion.\looseness=-1
\begin{proposition}\label{prop:safety}
Consider a cost function $c(\cdot)$ satisfying~\eqref{eq:c_cond}. If there exist constants $\beta, \epsilon\!\in\!\mathbb{R}_+$ and $\alpha\in[0,1]$ such that the barrier function \eqref{eq:barrier} satisfies 
\begin{align}\label{eq:safety_cond}
    -\mathbb{R}_{\beta}[-h(\bm{x}^+)]\geq \alpha h(\bm{x})+\epsilon, \qquad \forall\bm{x}\in\mathbb{V}_{\bm{\pi}}^{\xi}
\end{align}
for 
$\bm{x}^+\!=\!\bm{f}(\bm{x},\bm{\pi}(\bm{x}),\bm{\omega})$, then, $\bm{\pi}(\cdot)$ is $K$-step $\delta$-safe on $\mathbb{V}_{\bm{\pi}}^{\xi}$ with
\begin{align}\label{eq:delta}
    \delta = 1-(1\!-\!\exp(-\beta(\alpha h(\bm{x}_0)\!+\!\epsilon)))(1\!-\!\exp(-\beta\epsilon))^{K-1}.
\end{align}
\end{proposition}
\begin{proof}
    \cref{lem:set2cost} allows us to bound the probability of leaving $\mathbb{X}_{\mathrm{safe}}$ by the probability of leaving $\mathbb{V}_{\bm{\pi}}^{\xi}$. Thus, it suffices to derive an upper bound for $\mathbb{P}\left( V_{\bm{\pi}}(\bm{x}^+)> \xi\right)$ to show the probability of $1$-step safety. This probability can be expressed as
    \begin{align}\label{eq:safety_pf_1}
    \mathbb{P}\left( V_{\bm{\pi}}(\bm{x}^+)> \xi\right)=\mathbb{E}\left[I_{\xi}(V_{\bm{\pi}}(\bm{x}^+)\right],
    \end{align}
    where the indicator function $I_{\xi}:\mathbb{R}\rightarrow \{0,1\}$ is defined as $I_\xi(V)=0$ if $V\leq \xi$ and $I_\xi(V)=1$ otherwise.
    Note that the expectation in \eqref{eq:safety_pf_1} only affects $V_{\bm{\pi}}(\cdot)$ indirectly as it is a deterministic function evaluated at the random variable $\bm{x}^+$. Furthermore, the exponential function is strictly increasing and positive with $\exp(0)=0$, and $\beta$ is positive. Hence, the indicator function can be bounded by $I_{\xi}(V_{\bm{\pi}}(\bm{x}^+))\leq \exp\left(\beta\left( V_{\bm{\pi}}(\bm{x}^+)-\xi \right)\right).$
    This inequality yields
    $\mathbb{P}\left( V_{\bm{\pi}}(\bm{x}^+)> \xi\right)\leq \mathbb{E}\!\left[ \exp\left( \beta (V_{\bm{\pi}}(\bm{x}^+)-\xi) \right) \right]$
    by taking the expectation of both sides. By employing the definition of the risk operator, 
    the right  side can be simplified to 
    $\mathbb{P}\left( V_{\bm{\pi}}(\bm{x}^+)> \xi\right)\leq\exp\left(\beta  \mathbb{R}_{\beta}[V_{\bm{\pi}}(\bm{x}^+)-\xi]\right)$
    Using \eqref{eq:safety_cond} and $\alpha h(\bm{x})+\epsilon\geq \epsilon$ for $V_{\bm{\pi}}(\bm{x})\leq \xi$, this implies 
    $\mathbb{P}(\bm{x}^+\in\mathbb{V}_{\bm{\pi}}^{\xi}| \bm{x}\in\mathbb{V}_{\bm{\pi}}^{\xi})\geq 1-\exp\left(-\beta \epsilon\right)<1.$
    Moreover, we obtain a more specific bound given exact knowledge of the value function $V_{\bm{\pi}}(\bm{x}_0)$, i.e., $\mathbb{P}(\bm{x}^+\in\mathbb{V}_{\bm{\pi}}^{\xi}|V_{\bm{\pi}}(\bm{x}_0))\geq 1-\exp\left(-\beta (\alpha h(\bm{x}_0)+\epsilon)\right).$
    By chaining these conditional probabilities, we obtain the joint probability for a safety  violation, which implies~\eqref{eq:delta}.\looseness=-1
\end{proof}
This result intuitively extends \eqref{eq:h_cond_noisy} to the setting with uncertain system dynamics $\bm{f}(\cdot,\cdot,\cdot)$: Ignoring the minus signs on the left-hand side of \eqref{eq:safety_cond}, only the expectation in \eqref{eq:h_cond_noisy} is replaced by the risk operator. The similarities extend to an exponential decay of the guaranteed probability of safety in \eqref{eq:delta}, which can also be found in \cite{cosner_robust_2023}. While the risk-neutral formulation \eqref{eq:h_cond_noisy} can only increase this probability through an increase of the overall robustness via $\alpha$ and $\epsilon$, our risk-sensitive approach \eqref{eq:safety_cond} allows us to specifically target high uncertainty in the model via the risk parameter $\beta$.

Due to the definition of $h(\cdot)$ based on a value function $V_{\bm{\pi}}(\cdot)$, we do not necessarily need to manually specify a policy $\bm{\pi}(\cdot)$. Instead, we can formulate an optimization problem $\bm{\pi}_{\mathrm{safe}}=\min_{\pi}\mathbb{E}[V_{\bm{\pi}}(\bm{x})]$
to jointly determine $\bm{\pi}(\cdot)$ and $V_{\bm{\pi}}(\cdot)$. Under certain assumptions on the combination of cost $c(\cdot)$ and distribution of the system dynamics $\mathcal{F}$, the value function $V_{\bm{\pi}_{\mathrm{safe}}}(\cdot)$ can be shown to satisfy the conditions of \cref{prop:safety} \cite{lederer_risk-sensitive_2023}. Moreover, $\min_{\pi}\mathbb{E}[V_{\bm{\pi}}(\bm{x})]$
can be efficiently solved using off-the-shelf RL approaches in practice, such that barrier functions can be straightforwardly learned.\looseness=-1

\subsection{Distributed Safety Filters via Worst-Case Anticipation}\label{subsec:worst_case}

While we can straightforwardly obtain a barrier function using the approach presented in \cref{subsec:CBF_cent},
its direct usage in a  safety filter of the form \eqref{eq:safety_filter_default} requires knowledge of other agents' control inputs, which we assume unavailable. Thus, we cannot evaluate the safety condition \eqref{eq:safety_cond} in a distributed setting.
To overcome this limitation, we 
under-approximate the left side of \eqref{eq:safety_cond} so that its evaluation only requires information about a single agent's action. We achieve this by taking the minimum with respect to other agents' actions leading to the safety filter
\begin{subequations}\label{eq:safety_filter_pess}
    \begin{align}
    \bm{u}_{\mathrm{pess}}^i=&\min_{\bm{u}^i\in\mathbb{U}}\|\bm{u}^i-\bm{\pi}_{\mathrm{nom}}^i(\bm{x})\|^2\\
    &\text{s.t. } \min_{\mathclap{\substack{\bm{u}^j\in\mathbb{U}, j\neq i}}} ~ -\mathbb{R}_{\beta}[-h(\bm{x}^+_u)]\geq \alpha h(\bm{x})+\epsilon \label{eq:safety_const_pess}
    \end{align}
\end{subequations}
where $\bm{x}^+_u=\bm{f}(\bm{x},\bm{u},\bm{\omega})$.
Due to the minimum in \eqref{eq:safety_const_pess}, the safety filter pessimistically anticipates the worst actions of other agents making it robust against their particular choice. While this approach might be considered overly conservative, it should be noted that we ideally want safety filters to have no affect on the nominal policy for most of the safe subset $\mathbb{X}_{\mathrm{safe}}$. Hence, we expect these worst case actions to relevantly tighten the original constraint \eqref{eq:safety_cond} only in the proximity to unsafe states. 
The benefit of this pessimistic approach for a distributed safety filter is crucial: feasibility of \eqref{eq:safety_const_pess} for one agent immediately
ensures the satisfaction of \eqref{eq:safety_cond}  for the entire multi-agent.
\begin{proposition}\label{prop:safety_pess}
    Consider a MAS with unknown dynamics $\bm{f}(\cdot,\cdot,\cdot)$ satisfying \cref{ass:prob_model}. Moreover, assume that $\mathbb{V}_{\bm{\pi}_{\mathrm{safe}}}^{\xi}$ is included in $\mathbb{X}_{\mathrm{safe}}$. If there exist constants $\beta, \epsilon\in\mathbb{R}_+$ and $\alpha\in[0,1]$ such that the constraint \eqref{eq:safety_const_pess} is feasible at a state $\bm{x}\in\mathbb{V}_{\bm{\pi}_{\mathrm{safe}}}^{\xi}$ for some $i=1,\ldots,M$, any $\bm{u}_{\mathrm{pess},i}=[(\bm{u}^1)^T, \ldots, (\bm{u}_{\mathrm{pess}}^i)^T, \ldots, (\bm{u}^M)^T]$ defined via \eqref{eq:safety_filter_pess} guarantees the satisfaction of \eqref{eq:safety_cond} at~$\bm{x}$.
\end{proposition}
\begin{proof}
    We start the proof by considering a single agent~$i$ whose action is determined by \eqref{eq:safety_filter_pess}, while the actions of all other agents are set to arbitrary values $\tilde{\bm{u}}^j$, $j\neq i$. If \eqref{eq:safety_const_pess} is feasible at $\bm{x}\in\mathbb{V}_{\bm{\pi}_{\mathrm{safe}}}^{\xi}$, we can define
    \begin{align*}
        \!\bm{x}^+_{\tilde{u}_{\mathrm{pess}},i}&=\bm{f}\big(\bm{x},\big[(\tilde{\bm{u}}^1)^T,\ldots,(\bm{u}_{\mathrm{pess}}^i)^T,\ldots, (\tilde{\bm{u}}^M)^T\big]^T, \bm{\omega}\big),\!\\
        \!\bm{x}^+_{u_{\mathrm{pess},i}}&=\bm{f}\big(\bm{x},\big[(\bm{u}^1)^T,\ldots,(\bm{u}_{\mathrm{pess}}^i)^T,\ldots, (\bm{u}^M)^T\big]^T, \bm{\omega}\big).
    \end{align*}
    Then, it is straightforward to see that 
    \begin{align*}
        -\mathbb{R}_{\beta}[-h(\bm{x}^+_{ u_{\mathrm{pess}},i})] \geq \min_{\mathclap{\substack{\tilde{\bm{u}}^j\in\mathbb{R}^{d_u},\\
          \forall j=1,\ldots,M, j\neq i}}} ~ -\!\mathbb{R}_{\beta}[-h(\bm{x}^+_{\tilde{u}_{\mathrm{pess}},i})] \geq  \alpha h(\bm{x})\!+\!\epsilon
    \end{align*}
    holds due to the definition of $\bm{u}_{\mathrm{pess}}^i$ in \eqref{eq:safety_filter_pess}. Since this inequality holds for arbitrary $\bm{u}^j$, safety of the MAS is guaranteed irrespective of the other agents' chosen control inputs, i.e., for arbitrary $\bm{u}^j$, $j\neq i$.
\end{proof}
This result allows the straightforward extension of centralized safety-filters similar to \eqref{eq:safety_filter_default} to the distributed control setting without any changes to the construction of the CBF. This comes at the price that during the online evaluation of the filter \eqref{eq:safety_filter_pess}, feasibility from its centralized counterpart is not inherited.\looseness=-1 

\subsection{Distributed Safety Filters based on Control Proximity}\label{subsec:control_prox}

Since the worst-case safety filter \eqref{eq:safety_filter_pess} is not guaranteed to yield a control input for all and potentially not for any agent $i$, it remains to determine a back-up strategy that can be employed in case of infeasibility of \eqref{eq:safety_filter_pess}.
Under the assumption that the dynamics and the barrier function $h(\cdot)$ are sufficiently smooth, one reasonable strategy clearly relies in staying close to a safe policy $\bm{\pi}_{\mathrm{safe}}(\cdot)$, such that its safety guarantee can be inherited. \looseness=-1

In more detail, we 
assume that $\bm{\pi}_{\mathrm{safe}}(\cdot)$ satisfies the safety condition \eqref{eq:safety_cond} with some $\bar{\alpha}\in[0,1]$, $\bar{\epsilon}\in\mathbb{R}_{0,+}$. Then, we can choose values $\alpha<\bar{\alpha}$ and $\epsilon<\bar{\epsilon}$, for which we want to ensure the safety condition \eqref{eq:safety_cond} 
using our safety filter. We quantify the difference between both parameter choices via the safety margin\looseness=-1
\begin{align}\label{eq:Delta_fun}
    \Delta(\bm{x}) = (\bar{\alpha}-\alpha) h(\bm{x}) + \bar{\epsilon} -\epsilon.
\end{align}
We equally split this margin between the agents, which can be achieved using a safety filter
\begin{subequations}\label{eq:safety_filter_Lipschitz}
    \begin{align}
    \bm{u}_{\mathrm{prox}}^i=&\min_{\bm{u}^i}\|\bm{u}^i-\bm{\pi}_{\mathrm{nom}}^i(\bm{x})\|^2\\
    &\text{s.t. } \|\bm{u}^i-\bm{\pi}_{\mathrm{safe}}^i(\bm{x})\|\leq \frac{\Delta(\bm{x})}{ML_h L_{f_u^i}}\label{eq:safety_const_Lipschitz}
    \end{align}
\end{subequations}
if the dynamics are $L_{f_u^i}$-Lipschitz with respect to their second argument and the barrier function $h(\cdot)$ is $L_h$-Lipschitz. Since this safety filter ensures the proximity of the applied control inputs to $\bm{\pi}_{\mathrm{safe}}(\cdot)$, it is guaranteed to be feasible and inherits the safety guarantees as shown in the following proposition.

\begin{proposition}\label{prop:safety_Lip}
    Consider a MAS with unknown dynamics $\bm{f}(\cdot,\cdot,\cdot)$ satisfying \cref{ass:prob_model}, which has components $\bm{f}^i(\cdot,\cdot,\cdot)$ that admit Lipschitz constants $L_{f_u^i}\in\mathbb{R}_+$. Moreover, assume that $\mathbb{V}_{\bm{\pi}_{\mathrm{safe}}}^{\xi}$ is included in $\mathbb{X}_{\mathrm{safe}}$ and $h(\bm{x})=\xi-V_{\bm{\pi}_\mathrm{safe}}(\bm{x})$ is $L_h$-Lipschitz. If $\bm{\pi}_{\mathrm{safe}}(\cdot)$ satisfies \eqref{eq:safety_cond} with $\bar{\alpha}$, $\bar{\epsilon}$, and $\alpha<\bar{\alpha}$, $\epsilon<\bar{\epsilon}$ holds,
    then \eqref{eq:safety_const_Lipschitz} is feasible for all $i=1,\ldots, M$, and $\bm{u}_{\mathrm{prox}}=[(\bm{u}_{\mathrm{prox}}^1)^T, \ldots, (\bm{u}_{\mathrm{prox}}^M)^T]$ defined via \eqref{eq:safety_filter_Lipschitz} guarantees the satisfaction of \eqref{eq:safety_cond}.
\end{proposition}
\begin{proof}
    Lipschitz continuity of $h(\cdot)$ and $\bm{f}(\cdot,\cdot,\cdot)$ with respect to the second argument guarantees that 
    $|h(\bm{f}(\bm{x},\bm{\pi}_{\mathrm{safe}}(\bm{x}),\bm{\omega})) - h(\bm{f}(\bm{x},\bm{u},\bm{\omega}))| \leq
        \sum\nolimits_{i=1}^M L_h L_{f_u^i} \|\bm{\pi}_{\mathrm{safe}}^i(\bm{x})-\bm{u}^i\|$.
    Therefore, we can bound the risk of the value function 
    by
    $\mathbb{R}_{\beta}\left[-h(\bm{x}_u^+) \right]\geq \mathbb{R}_{\beta}\left[-h(\bm{x}_{\bm{u}'}^+)\right]-\!\sum\nolimits_{i=1}^M \! L_h L_{f_u^i} \|\bm{\pi}_{\mathrm{safe}}^i(\bm{x})\!-\!\bm{u}^i\|$
    with $\bm{u}'=\bm{\pi}_{\mathrm{safe}}(\bm{x}).$
    Since the safe back-up policy $\bm{\pi}_{\mathrm{safe}}(\cdot)$ ensures that \eqref{eq:safety_cond} is satisfied with $\bar{\alpha}, \bar{\epsilon}$, we have
    $\mathbb{R}_{\beta}\left[-h(\bm{x}_u^+) \right]\geq \bar{\alpha} h(\bm{x})+\bar{\epsilon}-\sum\nolimits_{i=1}^M L_h L_{f_u^i} \|\bm{\pi}_{\mathrm{safe}}^i(\bm{x})-\bm{u}^i\|.$
    Moreover, the constraint \eqref{eq:safety_const_Lipschitz} guarantees that $\sum_{i=1}^M L_h L_{f_u^i} \|\bm{\pi}_{\mathrm{safe}}^i(\bm{x})-\bm{u}^i\| \leq \Delta(\bm{x})$
    for $\Delta(\cdot)$ defined in \eqref{eq:Delta_fun}, which immediately implies \eqref{eq:safety_cond}.\looseness=-1
\end{proof}
In comparison to \cref{prop:safety_pess}, this result requires additional assumptions, but also provides stronger feasibility guarantees. However, the additionally required Lipschitz continuity is generally not very restricitve. 
Probabilistic estimates can be directly obtained from the probability distribution $\mathcal{F}$ over the dynamics, e.g., for Gaussian process models as shown in \cite{Lederer2019}. Lipschitz continuity of the barrier function $h(\cdot)$ follows immediately from a sufficient smoothness of $c(\cdot)$, $\bm{f}(\cdot,\cdot,\cdot)$ and $\bm{\pi}_{\mathrm{safe}}(\cdot)$, such that this assumption is also not very restrictive in general. 
Note that even in the absence of any knowledge of these Lipschitz constants, the result straightforwardly extends to the case $L_h,L_{f_u^i}\rightarrow \infty$, such that $\bm{u}_{\mathrm{prox}}^i=\bm{\pi}_{\mathrm{safe}}^i(\bm{x})$. Since the proximity constraint \eqref{eq:safety_const_Lipschitz} is always feasible,
\eqref{eq:safety_filter_Lipschitz} is also applicable when \eqref{eq:safety_filter_pess} becomes infeasible. However, this advantage comes at a price: All agents need to apply $\bm{u}^i_{\mathrm{prox}}$ for the safety condition \eqref{eq:safety_cond} to be satisfied. Moreover, the deviation of $\bm{u}^i_{\mathrm{prox}}$ from $\bm{\pi}_{\mathrm{safe}}^i(\bm{x})$ is bounded in terms of the Euclidean distance. This can potentially cause restrictions on the admissible control inputs in \eqref{eq:safety_filter_Lipschitz} even when far away from the boundary of the safe set $\mathbb{X}_{\mathrm{safe}}$. Therefore, its application can be undesirable in these situations.

\subsection{Feasibility-based Safety-Filter Switching}\label{subsec:merging}

Since the worst-case formulation \eqref{eq:safety_filter_pess} and the proximity-based approach \eqref{eq:safety_filter_Lipschitz} have complimentary strengths, we combine them to obtain a feasible safety filter, which does not exhibit excessive restrictiveness of the admissible control inputs. We achieve this by exploiting the property that the feasibility of \eqref{eq:safety_filter_pess} on one agent implies safety for the entire MAS. Thus, we define the safety-filtered policy of each agent as\looseness=-1
\begin{align}\label{eq:overall_filter}
    \bm{\pi}^i_{\mathrm{filt}}(\bm{x})=\begin{cases}
        \bm{u}_{\mathrm{pess}}^i&\text{if \eqref{eq:safety_const_pess} feasible}\\
        \bm{u}_{\mathrm{prox}}^i&\text{else, }
    \end{cases}
\end{align}
such that each agent can compute this safety filter and check feasibility of \eqref{eq:safety_const_pess} without coordination.
The computation of \eqref{eq:safety_filter_Lipschitz} in addition to the attempt of solving \eqref{eq:safety_filter_pess} only causes negligible computational overhead as \eqref{eq:safety_filter_Lipschitz} can be reformulated as an efficient second-order cone program.
As all agents independently finding \eqref{eq:safety_const_pess} infeasible implies $\bm{\pi}^i_{\mathrm{filt}}(\bm{x})=\bm{u}_{\mathrm{prox}}^i$ for all $i=1,\ldots,M$, safety can be ensured using \cref{prop:safety_Lip} when \cref{prop:safety_pess} is not applicable. 
Thereby, the safety and feasibility guarantees from \eqref{eq:safety_filter_pess} and  \eqref{eq:safety_filter_Lipschitz} are maintained.

\begin{theorem}
     Consider a MAS with unknown dynamics $\bm{f}(\cdot,\cdot,\cdot)$ satisfying \cref{ass:prob_model}, which has components $\bm{f}^i(\cdot,\cdot,\cdot)$ that admit Lipschitz constants $L_{f_u^i}\in\mathbb{R}_+$. Moreover, assume that $\mathbb{V}_{\bm{\pi}_{\mathrm{safe}}}^{\xi}$ is included in $\mathbb{X}_{\mathrm{safe}}$ and $h(\bm{x})=\xi-V_{\bm{\pi}_\mathrm{safe}}(\bm{x})$ is $L_h$-Lipschitz. If $\bm{\pi}_{\mathrm{safe}}(\cdot)$ satisfies \eqref{eq:safety_cond} with $\bar{\alpha}$, $\bar{\epsilon}$ for all $\bm{x}\in\mathbb{V}_{\bm{\pi}_{\mathrm{safe}}}^{\xi}$, and $\alpha<\bar{\alpha}$, $\epsilon<\bar{\epsilon}$ holds,
    then $\bm{\pi}^i_{\mathrm{filt}}(\cdot)$ is well-defined for all $\bm{x}\in\mathbb{V}_{\bm{\pi}_{\mathrm{safe}}}^{\xi}$ and guarantees $K$-step $\delta$-safety with $\delta$ defined in \eqref{eq:delta}.
\end{theorem}
\begin{proof}
    As $\bm{u}_{\mathrm{prox}}^i$ exists for all $\bm{x}\in\mathbb{V}_{\bm{\pi}_{\mathrm{safe}}}^{\xi}$ due to \cref{prop:safety_Lip}, $\bm{\pi}^i_{\mathrm{filt}}(\bm{x})$ is well-defined for all $\bm{x}\in\mathbb{V}_{\bm{\pi}_{\mathrm{safe}}}^{\xi}$. 
    Moreover, it follows from the proof of \cref{prop:safety_pess} that feasibility of \eqref{eq:safety_const_pess} for a single agent $i$ guarantees the satisfaction of \eqref{eq:safety_cond}. If all agents apply $\bm{u}_{\mathrm{prox}}^i$, the satisfaction of \eqref{eq:safety_cond} follows from \cref{prop:safety_Lip}. 
    Therefore, \eqref{eq:safety_cond} is satisfied for all $\bm{x}\in\mathbb{V}_{\bm{\pi}_{\mathrm{safe}}}^{\xi}$, such that \cref{prop:safety} guarantees $K$-step $\delta$-safety.
\end{proof}

\section{Numerical Evaluation}\label{sec:eval}

We evaluate the proposed distributed risk-sensitive safety filter in two examples. We first consider agents that are coupled in their dynamics as an abstraction for mobile robots jointly transporting an object in \cref{subsec:eval_dynCouple}. In \cref{subsec:eval_constCouple}, we investigate the behavior of the safety filter with growing number of agents that need to avoid collisions which each other. \looseness=-1

We use rewards 
$r(\bm{x}_k,\bm{u}_k)=\exp(-\|\bm{u}_k\|^2_{\bm{W}_u}-\sum_{i=1}^3 \|\bm{x}_{k}^i-\bm{x}_{\mathrm{ref}}\|^2_{\bm{W}_x^i})$ 
to specify the nominal policy, representing regulation to a reference state $\bm{x}_{\mathrm{ref}}$ as a prototypical task. Nominal and safe policies are computed using PPO with learning rate $1e\!-\!6$ and a discount factor $\gamma=0.99$ \cite{schulman_proximal_2017}. The obtained safe policy is used to generate $2e4$ training samples of the value function via roll-outs with $500$ time steps. 
Using these samples, the value function is learned using a neural network with $3$ hidden layers and $[256, 256, 64]$ neurons over 
$400$ episodes.
Other safety filter parameters are treated as design choices and we manually tune them by increasing their values starting from the theoretically required minimum until an acceptable performance is obtained. 
This yields $\alpha\!=\!0.1$, $\epsilon\!=\!0$, and $\nicefrac{\Delta(\bm{x})}{ML_hL_{f_u^i}}\!=\!0.05$. Note that we consider $\nicefrac{\Delta(\bm{x})}{ML_hL_{f_u^i}}$ as a single parameter. 
The expectations in the safety filters are approximated through empirical means with $5$ samples.

\begin{figure}
    \centering
    \includegraphics[]{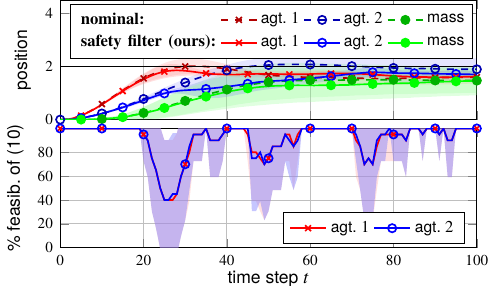}
    \caption{Top: Average trajectories for the nominal control law (dashed lines) overshoot and exhibit significant constraint violations, while our safety filter (full lines) slows down the increase to maintain safety. Shaded areas denote one standard deviation intervals. Bottom: The worst-case safety filter becomes mainly infeasible when one agent comes close to the constraint boundary, but has a high feasibility rate in general. This illustrates that switching in \eqref{eq:overall_filter} is limited to few situations.\looseness=-1 }
    \label{fig:sample_trajs}
\end{figure}

\subsection{Dynamics Coupling between Agents}
\label{subsec:eval_dynCouple}

We start with the problem of controlling two agents connected to a joint mass via a spring. Considering the mass as an unactuated third agent, we model this coupled system via $x_{k+1,1}^i=x_{k,1}^i+0.1x_{k,2}^i+\omega_{k,1}^i$, $x_{k+1,2}^i=x_{k,2}^i+0.1g^i(\bm{x}_k^1,\bm{u}_k) - 0.1\sin(\psi(x_{k,2}^i))+\omega_{k,2}^i$
where $\psi:\mathbb{R}\rightarrow [-1,1]$ is a saturating linear function,  $g^i(\bm{x}_k,\bm{u}_k)=5u^i-0.5\theta^2 e_k^i$ for $i=1,2$ and $g^3(\bm{x}_k,\bm{u}_k)=0.5\theta^2 (e_k^1+e_k^2)$ with $e_k^i = x_{k,1}^i-x_{k,1}^3$. We assume 
noise $\bm{\omega}_k\sim\mathcal{N}(\bm{0},\sigma_n^2\bm{I})$, $\sigma_n=0.01$, and a coupling parameter $\theta\sim\mathcal{N}(0,1)$ independently sampled in each roll-out. The goal is to steer the positions $x_{k,1}^i$ of all agents 
to a target value. For this purpose, we define the reward using $\bm{W}_u=0.01\bm{I}$, $\bm{W}_{x}^i=\mathrm{diag}(0.1,0)$, $i=1,2$, and $\bm{W}_{x}^3=\mathrm{diag}(1,0)$, and set $\bm{x}_{\mathrm{ref}}=[\nicefrac{7}{4}, 0]^T$. As a safety constraint, we consider the requirement $|x_{k,1}^i|\leq 2$ encoded through the cost function $c(\bm{x}_k)=1-\sum_{i=1}^3 \frac{1}{3}\mathrm{sigm}_{10}(2^2-(x_{k,1}^i)^2)$, where $\mathrm{sigm}_{10}(x)=\nicefrac{1}{(1+\exp(-10x))}$.\looseness=-1

The trajectories resulting from the nominal control law and our proposed safety filter with $\beta=1$ and $\xi=5$ averaged over $20$ roll-outs are illustrated at the top of \cref{fig:sample_trajs}. The nominal control law causes a significant overshoot for agents 1 and 2 to pull the unactuated mass to the target position, which causes a considerable number of constraint violations. In contrast, our proposed safety filter slows down the agents and thereby maintains constraint satisfaction with a high probability. As illustrated at the bottom of \cref{fig:sample_trajs}, this behavior is mainly achieved via the worst-case approach \eqref{eq:safety_filter_pess}, which does not cause excessive conservatism in this example. It only becomes infeasible when one of the agents approaches the constraint boundary, where our safety filter \eqref{eq:overall_filter} maintains safety by switching to the proximity-based approach \eqref{eq:safety_filter_Lipschitz}.

\begin{figure}
    \centering
    \includegraphics[]{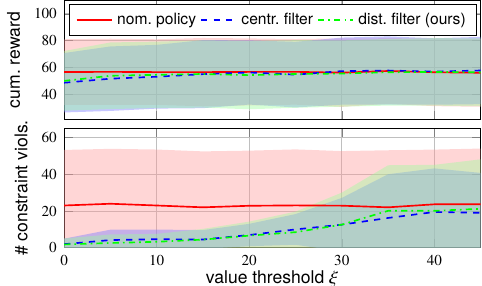}
    \caption{The cumulative reward (top) and number of constraint violations (bottom) for our distributed safety filter closely resemble the performance of the centralized baseline from \cite{lederer_risk-sensitive_2023}. While the rewards only deviate from the nominal policy for small safety thresholds $\xi$, a considerable impact on the number of constraint violations can still be observed for relatively large $\xi$. 
    Large standard deviation intervals (shaded areas) result from few roll-outs with negligible coupling $\theta\approx0$ causing a lack of controllability of the unactuated mass.
    }
    \label{fig:comparison}
\end{figure}

We additionally compare the performance and constraint satisfaction of our distributed safety filter to its centralized counterpart. As illustrated in \cref{fig:comparison}, our distributed safety filter yields results that closely resemble those of the centralized approach for a range of value function thresholds $\xi$. Moreover, both approaches cause a considerable deviation from the behavior of the nominal policy for small~$\xi$, while the amount of safety constraint violations continuously increases with growing thresholds. This observation stresses the effectiveness and flexibility of our distributed risk-sensitive safety filter for ensuring safety while limiting negative effects on performance.

\subsection{Agent Coupling via Collision Constraints}
\label{subsec:eval_constCouple}

In the second example, we consider $M$ independent agents $x_{k+1,1}^i=x_{k,1}^i+0.01 x_{k,2}^i + \theta\sin(x_{k,1}^i)+\omega_{k,1}^i$, $x_{k+1,2}^i=x_{k,2}^i+u_k^i+\omega_{k,2}^i$
with noise $\bm{\omega}_k^i\sim\mathcal{N}(0,\sigma_n^2I)$, $\sigma_n=0.1$, and parameter $\theta\sim\mathcal{N}(0,1)$ independently sampled in each roll-out. The nominal task of the agents is to move to the origin, for which we define the reward using $\bm{W}_u=0.1\bm{I}$,  $\bm{W}_x^i=\mathrm{diag}(1,0.1)$, and $\bm{x}_{\mathrm{ref}}=\bm{0}$. While the agents are not coupled in the dynamics, we couple them through collision constraints $\|\bm{x}_{k,1}^i-\bm{x}_{k,1}^j\|\geq 0.2$ for all $i\neq j$. We encode these constrains via the cost $c(\bm{x}_k)=\frac{1}{M}\sum_{i=1}^M \mathrm{sigm}_{10}(0.04-\min_j e_{k,j}^i)$ where $e_{k,j}^i=(x_{k,1}^i-x_{k,1}^j)^2$.

\begin{figure}
    \centering
    \includegraphics[]{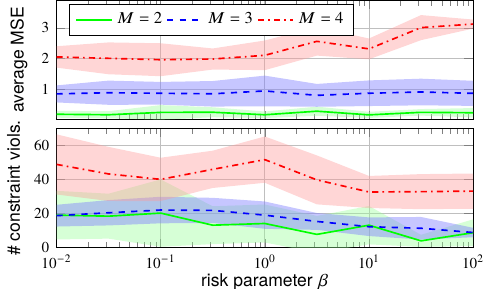}
    \caption{Top: While the tracking error is barely affected by the risk-aversion for $M=2$ and $M=3$, the increased difficulty of ensuring safety for $M=4$ causes a growth of the MSE with growing~$\beta$. Bottom: Irrespective of the number of agents, the number of safety violations reduces with growing risk-aversion $\beta$. Shaded areas denote one standard deviation intervals.\looseness=-1
    }
    \label{fig:agent_dep}
\end{figure}

For this system, we investigate the dependency of the proposed distributed safety filter on the number of agents and the risk parameter $\beta$ with randomly initialized trajectories. As illustrated at the bottom of \cref{fig:agent_dep}, the number of safety violations averaged over $10$ roll-outs decreases with growing risk aversion $\beta$ irrespective of the number of agents\footnote{We use $1e5$ training samples to learn the value function for $M=4$ to account for the higher state dimensionality with 4 agents.}. Thereby, it validates the behavior of the probability of safety guaranteed by \cref{prop:safety}. 
Since the limit cases $\beta\rightarrow 0$ and $\beta\rightarrow \infty$ correspond to the risk-neutral and robust safety filter \cite{286253}, respectively, a close approximation of the performance of common baselines can be seen at the edges of \cref{fig:agent_dep}. This illustrates the flexibility of risk sensitivity enabling the interpolation between these extremes.
Moreover, we can see a growth of the number of constraint violations with an increasing number of agents $M$. This effect can be attributed to the growing likelihood of safety violations for random initial states in combination with a significant growth of the problem difficulty from 3 to 4 agents. 
Note that our proposed distributed safety filter is generally limited to small-scale MASs due the required knowledge of the overall system state and the quickly growing complexity of the optimization problems in the pessimistic formulation \eqref{eq:safety_filter_pess}.
As depicted at the top of \cref{fig:agent_dep}, a similar behavior can be observed for the mean squared error (MSE), which only increases for $M=4$. However, it effectively remains constant for 2 and 3 agents which highlights the benefits of our risk-averse approach: A risk-based handling of uncertainties can allow to improve safety while barely affecting the performance.\looseness=-1

\section{Conclusion}\label{sec:conc}

We proposed a distributed risk-sensitive safety filter for uncertain MASs using CBFs that are defined via value functions. Using risk-averse CBF conditions and switching between different filtering strategies, we obtain a feasible and distributed formulation for ensuring safety. Detailed numerical evaluations demonstrate the effectiveness and flexibility of our approach. In future work, we will explore distributed optimization methods for allowing a decentralized training of the CBFs. Moreover, we will develop variants of our approach that ensure safety without access to the full state information of all agents to increase its scalability.

\bibliographystyle{IEEEtran}
\bibliography{refs.bib}

\begin{thebibliography}{10}
\providecommand{\url}[1]{#1}
\csname url@samestyle\endcsname
\providecommand{\newblock}{\relax}
\providecommand{\bibinfo}[2]{#2}
\providecommand{\BIBentrySTDinterwordspacing}{\spaceskip=0pt\relax}
\providecommand{\BIBentryALTinterwordstretchfactor}{4}
\providecommand{\BIBentryALTinterwordspacing}{\spaceskip=\fontdimen2\font plus
\BIBentryALTinterwordstretchfactor\fontdimen3\font minus \fontdimen4\font\relax}
\providecommand{\BIBforeignlanguage}[2]{{%
\expandafter\ifx\csname l@#1\endcsname\relax
\typeout{** WARNING: IEEEtran.bst: No hyphenation pattern has been}%
\typeout{** loaded for the language `#1'. Using the pattern for}%
\typeout{** the default language instead.}%
\else
\language=\csname l@#1\endcsname
\fi
#2}}
\providecommand{\BIBdecl}{\relax}
\BIBdecl

\bibitem{gronauer_multi-agent_2022}
S.~Gronauer and K.~Diepold, ``Multi-{Agent} {Deep} {Reinforcement} {Learning}: {A} {Survey},'' \emph{Artificial Intelligence Review}, vol.~55, no.~2, pp. 895--943, 2022.

\bibitem{wabersich_data-driven_2023}
K.~P. Wabersich, A.~J. Taylor, J.~J. Choi, K.~Sreenath, C.~J. Tomlin, A.~D. Ames, and M.~N. Zeilinger, ``Data-{Driven} {Safety} {Filters}: {Hamilton}-{Jacobi} {Reachability}, {Control} {Barrier} {Functions}, and {Predictive} {Methods} for {Uncertain} {Systems},'' \emph{IEEE Control Systems Magazine}, vol.~43, no.~5, pp. 137--177, 2023.

\bibitem{Ames2017}
A.~D. Ames, X.~Xu, J.~W. Grizzle, and P.~Tabuada, ``Control {Barrier} {Function} {Based} {Quadratic} {Programs} for {Safety} {Critical} {Systems},'' \emph{IEEE Transactions on Automatic Control}, vol.~62, no.~8, pp. 3861--3876, 2017.

\bibitem{cohen2024safety}
M.~H. Cohen, T.~G. Molnar, and A.~D. Ames, ``Safety-critical control for autonomous systems: Control barrier functions via reduced-order models,'' \emph{Annual Reviews in Control}, vol.~57, p. 100947, 2024.

\bibitem{agrawal_discrete_2017}
A.~Agrawal and K.~Sreenath, ``Discrete {Control} {Barrier} {Functions} for {Safety}-{Critical} {Control} of {Discrete} {Systems} with {Application} to {Bipedal} {Robot} {Navigation},'' in \emph{Robotics: {Systems} and {Science}}, 2017.

\bibitem{so_how_2024}
O.~So, Z.~Serlin, M.~Mann, J.~Gonzales, K.~Rutledge, N.~Roy, and C.~Fan, ``How to {Train} {Your} {Neural} {Control} {Barrier} {Function}: {Learning} {Safety} {Filters} for {Complex} {Input}-{Constrained} {Systems},'' in \emph{Proceedings of the {IEEE} {International} {Conference} on {Robotics} and {Automation}}, 2024, pp. 11\,532--11\,539.

\bibitem{Curi2022}
S.~Curi, A.~Lederer, S.~Hirche, and A.~Krause, ``Safe {Reinforcement} {Learning} via {Confidence}-{Based} {Filters},'' in \emph{Proceedings of the {IEEE} {Conference} on {Decision} and {Control}}, 2022, pp. 3409--3415.

\bibitem{cosner_robust_2023}
R.~K. Cosner, P.~Culbertson, A.~J. Taylor, and A.~D. Ames, ``Robust {Safety} under {Stochastic} {Uncertainty} with {Discrete}-{Time} {Control} {Barrier} {Functions},'' in \emph{Robotics: {Science} and {Systems}}, 2023.

\bibitem{Ahmadi2022}
M.~Ahmadi, X.~Xiong, and A.~D. Ames, ``Risk-{Averse} {Control} via {CVaR} {Barrier} {Functions}: {Application} to {Bipedal} {Robot} {Locomotion},'' \emph{IEEE Control Systems Letters}, vol.~6, pp. 878--883, 2022.

\bibitem{lederer_risk-sensitive_2023}
A.~Lederer, E.~Noorani, J.~S. Baras, and S.~Hirche, ``Risk-{Sensitive} {Inhibitory} {Control} for {Safe} {Reinforcement} {Learning},'' in \emph{Proceedings of the {IEEE} {Conference} on {Decision} and {Control}}, 2023, pp. 1040--1045.

\bibitem{GARG2024100948}
K.~Garg, S.~Zhang, O.~So, C.~Dawson, and C.~Fan, ``Learning safe control for multi-robot systems: Methods, verification, and open challenges,'' \emph{Annual Reviews in Control}, vol.~57, p. 100948, 2024.

\bibitem{ahmadi_safe_2019}
M.~Ahmadi, A.~Singletary, J.~W. Burdick, and A.~D. Ames, ``Safe {Policy} {Synthesis} in {Multi}-{Agent} {POMDPs} via {Discrete}-{Time} {Barrier} {Functions},'' in \emph{Proceedings of the {IEEE} {Conference} on {Decision} and {Control}}, 2019, pp. 4979--4803.

\bibitem{cheng_safe_2020}
R.~Cheng, M.~J. Khojasteh, A.~D. Ames, and J.~W. Burdick, ``Safe {Multi}-{Agent} {Interaction} through {Robust} {Control} {Barrier} {Functions} with {Learned} {Uncertainties},'' in \emph{Proceedings of the {IEEE} conference on {Decision} and {Control}}, 2020, pp. 777--783.

\bibitem{zhang_discrete_2025}
S.~Zhang, O.~So, M.~Black, and C.~Fan, ``Discrete {GCBF} {Proximal} {Policy} {Optimization} for {Multi}-agent {Safe} {Optimal} {Control},'' in \emph{Proceedings of the {International} {Conference} on {Learning} {Representations}}, 2025.

\bibitem{Rasmussen2006}
C.~E. Rasmussen and C.~K.~I. Williams, \emph{Gaussian {Processes} for {Machine} {Learning}}.\hskip 1em plus 0.5em minus 0.4em\relax Cambridge, MA: The MIT Press, 2006.

\bibitem{rothfuss_bridging_2024}
J.~Rothfuss, B.~Sukhija, L.~Treven, F.~Dörfler, S.~Coros, and A.~Krause, ``Bridging the {Sim}-to-{Real} {Gap} with {Bayesian} {Inference},'' in \emph{Proceedings of the IEEE/RSJ International Conference on Intelligent Robots and Systems}, 2024, pp. 10\,784--10\,791.

\bibitem{286253}
M.~James, J.~Baras, and R.~Elliott, ``{Risk-Sensitive Control and Dynamic Games for Partially Observed Discrete-Time Nonlinear Systems},'' \emph{IEEE Transactions on Automatic Control}, vol.~39, no.~4, pp. 780--792, 1994.

\bibitem{Lederer2019}
A.~Lederer, J.~Umlauft, and S.~Hirche, ``Uniform~{Error}~{Bounds}~for~{Gaussian}~{Process} {Regression} with {Application} to {Safe} {Control},'' in \emph{Advances~in~{Neural}~{Information}~{Processing}~{Systems}}, 2019, pp. 659--669.

\bibitem{schulman_proximal_2017}
\BIBentryALTinterwordspacing
J.~Schulman, F.~Wolski, P.~Dhariwal, A.~Radford, and O.~Klimov, ``Proximal {Policy} {Optimization} {Algorithms},'' 2017. [Online]. Available: \url{http://arxiv.org/abs/1707.06347}
\BIBentrySTDinterwordspacing

\end{thebibliography}

\end{document}